\documentclass{IEEEtaes}
\usepackage{indentfirst}
\usepackage{amsmath}
\usepackage[final]{graphicx}
\usepackage{dsfont}
\usepackage{amsfonts}
\usepackage{amssymb}
\usepackage{bm}
\usepackage{subfigure}
\usepackage{booktabs}
\usepackage{graphicx}
\usepackage{caption}
\usepackage{epstopdf}
\usepackage{CJK}
\usepackage{cite}
\usepackage{pifont}
\usepackage{hyperref}
\usepackage{color,array,amsthm}
\usepackage{graphicx}

\jvol{XX}
\jnum{XX}
\jmonth{XXXXX}
\paper{1234567}
\pubyear{2022}
\doiinfo{TAES.2022.Doi Number}

\newcommand{\ba}{\begin{array}}
\newcommand{\ea}{\end{array}}
\newcommand{\bena}{\begin{eqnarray}}
\newcommand{\eena}{\end{eqnarray}}

\newcommand{\ec}{\end{center}}

\newcommand{\ei}{\end{itemize}}
\newcommand{\benu}{\begin{enumerate}}
\newcommand{\eenu}{\end{enumerate}}
\newcommand{\bdes}{\begin{description}}
\newcommand{\edes}{\end{description}}

\newcommand{\et}{\end{tabular}}

\newcommand \deltabf{\mbox{\boldmath$\delta$\unboldmath}}

\newcommand{\circlambda}{\mbox{$\Lambda$
             \kern-.85em\raise1.5ex
             \hbox{$\scriptstyle{\circ}$}}\,}

\renewcommand \deltabf{\boldsymbol{\delta}}

\newtheorem{proposition}{Proposition}

\def\cC{\mbox{$\mathcal C$}}
\def\cC{\mbox{$\mathcal C$}}

\def\cL{\mbox{$\mathcal L$}}
\def\cN{\mbox{$\mathcal N$}}
\def\cN{\mbox{$\mathcal N$}}

\def\cW{\mbox{$\mathcal W$}}

\newcommand{\bbm}{\begin{bmatrix}}
\newcommand{\ebm}{\end{bmatrix}}
\newcommand{\bit}{\begin{itemize}}
\newcommand{\eit}{\end{itemize}}
\newcommand{\ben}{\begin{enumerate}}
\newcommand{\een}{\end{enumerate}}
\newcommand{\bdesc}{\begin{description}}
\newcommand{\edesc}{\end{description}}
\newcommand{\bea}{\begin{array}}
\newcommand{\eea}{\end{array}}
\newcommand{\beqa}{\begin{eqnarray}}
\newcommand{\eeqa}{\end{eqnarray}}
\newcommand{\be}{\begin{equation}}
\newcommand{\ee}{\end{equation}}

\newcommand{\Comment}[1]{}

\newcommand{\boe}{{\mbox{\boldmath $e$}}}

\newcommand{\bn}{{\mbox{\boldmath $n$}}}
\newcommand{\bp}{\mbox{\boldmath $p$}}

\newcommand{\bv}{{\mbox{\boldmath $v$}}}

\newcommand{\bz}{{\mbox{\boldmath $z$}}}

\newcommand{\bA}{{\mbox{\boldmath $A$}}}

\newcommand{\bI}{{\mbox{\boldmath $I$}}}

\newcommand{\bM}{{\mbox{\boldmath $M$}}}

\newcommand{\bS}{{\mbox{\boldmath $S$}}}

\newcommand{\bU}{{\mbox{\boldmath $U$}}}

\newcommand{\bX}{{\mbox{\boldmath $X$}}}

\newcommand{\bZ}{{\mbox{\boldmath $Z$}}}

\newcommand{\btheta}{{\mbox{\boldmath $\theta$}}}

\newcommand{\bnu}{{\mbox{\boldmath $\nu$}}}

\newcommand{\testinv}{\mathop{\gtrless}\limits^{H_1}_{H_0}}

\newcommand{\testp}{\mathop{\gtrless}\limits^{H_1}_{H_0}}

\IEEEoverridecommandlockouts
\begin{document}

\title{{\color{black}Joint ML-Bayesian Approach to Adaptive Radar Detection in the presence of Gaussian Interference}}

\author{Chaoran Yin},
\author{Tianqi Wang},
\author{Linjie Yan},
\author{Chengpeng Hao},
\member{Senior Member, IEEE}
\affil{Institute of Acoustics,

Chinese
Academy of Sciences, Beijing, 100190, China}

\author{Alfonso Farina}
\member{Life Fellow, IEEE}
\affil{Selex ES (retired), a consultant, Roma, Italy}

\author{Danilo Orlando}
\member{Senior Member, IEEE}
\affil{Dipartimento di Ingegneria dell'Informazione,
Universita di Pisa, Via Caruso 16, 56122 Pisa, Italy}

\receiveddate{This work was supported
by the National Natural Science Foundation of China under Grant 62201564.
The work of D.Orlando was partially supported by the Italian Ministry
of Education and Research (MUR) in the framework of the FoReLab project
(Departments of Excellence) and in part by the European Union in the
NextGenerationEU plan through the Italian program ``Bando PRIN 2022'',
D.D. 104/2022 (PE7, project ``CIRCE,'' code H53D23000420006).}

\corresp{ {\itshape (Corresponding author: Chengpeng Hao)}. }

\authoraddress{
Chaoran Yin, Tianqi Wang, Linjie Yan, and Chengpeng Hao are with the Institute of Acoustics, Chinese
Academy of Sciences, Beijing
100190, China, (e-mail: \href{mailto:yinchaoran18@mails.ucas.ac.cn}{yinchaoran18@mails.ucas.ac.cn},
\href{mailto:wangtianqi@mail.ioa.ac.cn}{wangtianqi@mail.ioa.ac.cn},
\href{mailto:yanlinjie16@163.com}{yanlinjie16@163.com},
\href{mailto:haochengp@mail.ioa.ac.cn}{haochengp@mail.ioa.ac.cn});
Alfonso Farina is with Selex ES (retired), a consultant, Roma, Italy.
(e-mail: \href{mailto:alfonso.farina@outlook.it}{alfonso.farina@outlook.it});
Danilo Orlando is with Dipartimento di Ingegneria dell'Informazione,
Universita di Pisa, Via Caruso 16, 56122 Pisa, Italy. (e-mail: \href{mailto:danilo.orlando@unipi.it}{danilo.orlando@unipi.it});
}


\maketitle

\begin{abstract}This paper addresses the adaptive radar target detection problem in the presence of
Gaussian interference with unknown statistical properties. To this end, the
problem is first formulated as a binary hypothesis test, and then we derive
a detection architecture grounded on the {\color{black} hybrid of Maximum Likelihood (ML) and Maximum A Posterior (MAP)
approach}. Specifically, we resort to the hidden discrete latent variables in conjunction with the Expectation-Maximization (EM) algorithms which cyclically updates the estimates of the unknowns.
In this framework, the estimates of the {\em a posteriori} probabilities under each hypothesis
are representative of the inherent nature of data and
used to decide for the presence of a potential target.
In addition, we prove that the developed detection scheme
ensures the desired Constant False Alarm Rate property with respect to the unknown interference
covariance matrix.
Numerical examples obtained through synthetic and real recorded data corroborate the effectiveness of the proposed
architecture and show that the MAP-based approach ensures evident improvement with respect to the
conventional generalized likelihood ratio test at least for the considered scenarios and parameter setting.
\end{abstract}


\section{Introduction}
I{\scshape n} the course of the last century, radar systems are becoming increasing ubiquitous
and affect many aspects of real life. As a consequence,
adaptive radar target detection in the presence
of Gaussian background with unknown statistical properties has been addressed for decades \cite{HaoSpringer}.
Based on the well-known Neyman-Pearson criterion \cite{KayBook},
a plethora of solutions can be found in the open literature, especially
detection architectures devised under increasingly challenging working
scenarios due to the development of electronic technology in the last years.
In the seminal paper by Kelly \cite{Kelly1986GLRT}, the related binary hypothesis testing
problem is solved by means of the Generalized
Likelihood Ratio Test (GLRT), which consists in comparing likelihood ratio
to a threshold corresponding to a pre-assigned False Alarm Probability ($P_{fa}$);
the unknown parameters are replaced by their maximum likelihood estimates. The
GLRT enjoys the desired Constant False Alarm Rate (CFAR) property with respect to the background
and uses a set of secondary data, free of target components and sharing
the same interference characterization as that in the Cell Under Test (CUT).
An {\em ad hoc} modification of the GLRT is devised in \cite{AMF1992}
under the assumption that the covariance of the Gaussian interference is perfectly known and then replaced
by the sample covariance matrix formed using the secondary data. The Adaptive
Matched Filter (AMF) developed in \cite{AMF1992} reduces significantly the computational load
at the cost of a slight degradation of the
detection performance with respect to the GLRT. This result inspired a continuous investigation
of several approaches in the radar community also considering that there do not exist
Uniformly Most Powerful (UMP) detectors for this problem \cite{DeMaio-KA2010}.

As a matter of fact, the detection performance of the GLRT and AMF heavily
depends on the quality of the estimates of Clutter Covariance Matrix (CCM).
In this respect, the knowledge-aided framework has been proposed which consists in
exploiting prior knowledge about the systems and/or
background at the design stage.
For example, the symmetry of the antenna array or pulse trains
is widely considered, which leads to a persymmetric CCM that is
Hermitian with respect to the diagonal and symmetric with respect to the counter-diagonal \cite{Nitzberg1980persym,GaoPACE,Hao2014Pers,PaillousPer1,GINOLHAC2014242,Cai1992,Liu2015-1}.
This specific structure greatly reduces the unknown parameters in the matrix
and thus allows for enhancement of the estimation quality for a given number of secondary data.
As an alternative approach, symmetric spectrum detectors are derived by exploiting the CCM in real domain
that arises when the power spectrum density of the received clutter is symmetric with respect to zero
after compensation for the Doppler frequency \cite{Foglia2017Symmetry,DeMaioSymmetric,Hao-ICASSP16}.
It is important to notice that the above knowledge-aided detectors enjoy
reliable detection performance in the {\em sample-starved} situations \cite{Foglia2017joint,Hao2016Joint}
where the volume of uniform secondary data is limited and the conventional detectors significantly
impair their performance \cite{RMB}.

Another approach used to capitalize the target energy consists in taking into account
the fact that target components are present in adjacent range/Doppler bins (spillover)
\cite{Spillover2011Danilo}. This energy contamination is due to the mismatch between the
sampling point and the true position of the
target and, consequently, the latter straddles between two consecutive cells,
possibly causing significant signal power loss in the CUT.
In \cite{Spillover2011Danilo}, the authors first develop the spillover architecture using continuous range bins and
come up with three detectors based on one-step and two-step GLRT
that are capable of joint detection and localization of a target within the CUT.
It has been verified that when the mismatch occurs, the detectors devised under the
assumption outperform the counterparts and provide accurate target
range estimation within a given range cell \cite{Hao2015Persy,Yan2017Symmetric,YLJ2020Parametric,Maio2014}.
This architecture is extended to the range-Doppler plane for Multi-Input
Multi-Output radar in \cite{WangTQ2024-2Dspillover}, wherein an {\em ad hoc} approximation of the
GLRT is adopted due to complex mathematics.
It is also worth mentioning that knowledge-aided spillover detectors are
obtained in \cite{Hao2015Persy,Yan2017Symmetric} to further enhance the
robustness of the systems.

Subspace signal model represents another viable mean to deal with the mismatch
and target energy loss in radar operating scenarios \cite{FarinaOTH2003,Danilo2022Subspace-I}. In fact, such paradigm
arises from the desire to cope with the uncertainty
related to the radar array steering direction caused by, for example, unbalanced channels, miscalibration
errors, and mutual coupling \cite{VT2002,Klemm2002,Bandiera2009}, and mitigate the energy degradation in the presence
of signal mismatches. In this framework, the uncertainty is accounted for by the unknown
coordinates of a signal that lies in a subspace
(perfectly known or known only by the dimension)
\cite{DD,8082118,9140357,Liuw-DoubleSpaceII,Ciuonzo2016StructuredInter-I,Ciuonzo2016StructuredInter-II, RicciSubspace,XU201879}.
In \cite{Danilo2022Subspace-I} and \cite{Danilo2022Subspace-II}, the authors derive several detectors
based on GLRT for both unconstrained coordinates (first-order) and constrained coordinates by
a prior distribution (second-order),
and accordingly establish a unified framework for adaptive subspace detection.
The subspace detectors are desirable when radar is working on the searching mode\cite{Bandiera2009}
due to a certain level of tolerance for target mismatch (robustness).
On the other hand, selective detectors such as Rao test \cite{Maio-Rao,Orlando-RAO},
which are sensitive to the mismatch and tend to reject the target misaligned with the nominal steering
vector, are preferable in target localization.
In \cite{Angelo2020CFAR}, the problem of designing robust or selective detectors is addressed
by deriving the decision region boundary of the clustered data in a suitable plane formed
by maximal invariants \cite{Invariant1995}. In particular, it is shown that
the selective detector can be obtained without sacrificing detection performance under matched
conditions.

Observe that the above detection architectures share the common assumption that the
unknown parameters are considered deterministic. As an alternative, the Bayesian approach
consists in assuming that the unknowns are random variables and obey {\em a priori}
distributions \cite{Murphy2012,Svensson2005IWishart,4359546,LI20191}. Such a line of reasoning leads to the so-called {\em a posteriori}
distributions and has been widely applied in radar target detection. For example, \cite{DeMaio-KA2010}
assumes that the CCM is inverse Wishart distributed and develops three Bayesian detectors based on the GLRT.
As a consequence, the prior knowledge is included via a colored-loading and greatly improves
the parameter estimation performance.
The effectiveness is corroborated through extensive examples conducted on simulation and real recorded databases.
In \cite{WangPu2011BayesianPara}, Bayesian parametric detectors are devised by assuming
{\em a priori} distribution for the noise in an auto-regressive process.
For cognitive radar \cite{7906008}, Bayesian method is an alternative in the decision-action phase which
adaptively optimizes the transmitted waveform according to the knowledge perceived from the background \cite{8398580}.
Furthermore, as indicated in \cite{Yin2022learningSPL}, the Bayesian approach can also be used to
learn the specific structure of the unknown CCM, which
is a preliminary step for applying the knowledge-aided detectors mentioned before.
Given the advantages of the Bayesian philosophy, \cite{MAPRatioTest2021}
proposes a Maximum A Posterior(MAP) test for target detection which compares the ratio
of {\em a posteriori} distributions maximized over joint parameter and hypothesis space
with a threshold set according to the false alarm rate. Illustrative examples show that
the detector performs significantly better than the GLRT-based detectors.

{\color{black}
Therefore, recalling that a UMP test for the fundamental
problem defined in \cite{Kelly1986GLRT} does not exist and encouraged
by the aforementioned efforts towards more and more
powerful decision rules, in this paper, we provide an
innovative solution to the conventional adaptive detection
problem by resorting to a joint Bayesian and Maximum Likelihood (ML) approach.
Unlike \cite{MAPRatioTest2021}, we do not require any {\em a priori} information about the
unknown parameters and the {\em a posteriori} probabilities are
iteratively obtained by applying the Expectation-Maximization (EM)
algorithm \cite{EM1977,YLJ2023Hete,HeteroEM2024anglo,Pia2021Learning}. Specifically, we introduce a fictitious
discrete random variable \cite{Murphy2012} that represents the status
of the CUT for what concerns the target presence. This
random variable is not observable and its value assigns
a label (or class) to the CUT. Thus, we modify the original data
distribution in order to account for this hidden random
variable and estimate the corresponding unknown
parameters through the EM. This line of reasoning allows us
to obtain the estimates of the {\em a posteriori} probabilities
that the hidden class takes on a value given
the data from the CUT. Such probabilities can then
be used to devise a
decision scheme grounded on the minimization of the Bayes risk and that
exploits the ML-based estimates returned by the EM algorithm.
Remarkably, we demonstrate that
the proposed decision scheme ensures the
CFAR property with respect to the interference covariance matrix.
The main contributions of this paper are
\begin{itemize}
\item the design of a joint Bayesian-ML detection architecture that
allows for a control of the $P_{fa}$ through a suitable
parameter. Notice also that the proposed approach
can be considered as a generalization of the MAP design criterion;
\item the application of the EM algorithm to estimate the unknown parameters
used to build up the decision scheme;
\item a formal proof of the CFAR property of the
proposed detector showing that the invariance with
respect to the interference covariance matrix holds true
when the iterative procedure starts from suitable initial values.
\end{itemize}
}

The remainder of the paper is organized as follows. The next section is devoted to the
problem formulation, whereas the Section III contains the detection architecture design and parameter estimation procedures.
Section IV analytically demonstrates the CFAR property of the proposed architecture, while illustrative examples and the related discussion are given in Section V.
Finally, in Section VI, we draw some concluding remarks and present possible future work routes.

\subsection{Notation}
In what follows, vectors (matrices) are denoted by boldface lower (upper) case letter.
Superscripts $(\cdot)^T$ and $(\cdot)^\dag$ denote transpose and complex conjugate transpose, respectively.
$\bI_N$ stands for an identity matrix of size $N\times N$.
$\bnu(\bM)$ is the vector containing the general distinct entries of the matrix $\bM$.
The notation $\sim$ means ``is distributed as''. $\cC\cN_N(\bm{\mu},\bX)$
denotes the $N$-dimensional circular complex Gaussian distribution with mean $\bm{\mu}$ and
positive-definite covariance matrix $\bX$. $P(c=m)$ stands for the probability that a discrete
random variable $c$ equals to $m$.
$\cC\cW(K,N;\bI_N)$
denotes the $N$-dimensional complex Wishart distribution with $K$ degrees of freedom and scale matrix $\bI_N$.
$\det(\cdot)$ and $\mbox{Tr}(\cdot)$ represent the determinant and trace of a matrix, respectively.
$\left|\cdot\right|$ denotes the modulus of a scalar.
The symbol $\bf 0$ denotes a matrix or vector of zeros of proper dimensions.
As for the numerical sets, $\mathds{C}$ is the set of complex numbers and $\mathds{C}^{N\times M}$ is the Euclidean space
of ($N\times M$)-dimensional complex matrices (or vectors, if $M=1$).

\section{Problem Formulation and EM-based Estimation}

Let us assume a radar system equipped with $N\geq 2$ time and/or space channels.
The echoes of the transmitted waveform from the region of interest are properly
collected, sampled, and processed to form several $N$ dimensional vectors representative of consecutive range bins. We denote by $\bz\in\mathds{C}^{N\times 1}$ the signal
vector corresponding to the CUT and we want to decide whether or not it contains target component.
To this end, the signal vectors in the neighbourhood of the CUT, denoted by $\bz_1,\ldots,\bz_K\in\mathds{C}^{N\times 1}$,
are utilized as secondary data \footnote{\color{black}{Notice that secondary data $\bz_k$s are assumed free of target components regardless which hypothesis is in force.}}
to estimate the interference statistical properties.
Then, the radar target detection problem can be formulated as a binary hypothesis test, i.e.,
\begin{equation}\label{hypotheses}
\left\{
\begin{array}{lll}
H_0:\left\{
\begin{array}{lll}
\bz = \bn,\\
\bz_k = \bn_k, k=1,\ldots,K,
\end{array}
\right.\\
H_1:\left\{
\begin{array}{lll}
\bz = \alpha\bv+ \bn,\\
\bz_k = \bn_k, k=1,\ldots,K,
\end{array}
\right.
\end{array}
\right.
\end{equation}
where $\bn,\bn_k, k=1,\ldots,K,$ are the interference component modeled as statistically independent circular complex
Gaussian distributed variables with zero mean and covariance matrix $\bM\in\mathds{C}^{N\times N}$,
$\alpha$ stands for the unknown complex amplitude of target echo, and $\bv\in\mathds{C}^{N\times 1}$ is the
nominal steering vector. Denote by $\bZ = [\bz,\bz_1,\ldots,\bz_K]\in\mathds{C}^{N\times (K+1)}$ the overall data matrix, then the Probability Density Function (PDF) of $\bZ$ under $H_0$ and $H_1$ is given by
\begin{equation}
f_0(\bZ;\btheta_0) = \left[\frac{1}{\pi^N\det(\bM)}\right]^{(K+1)}\exp{\left\{-\mbox{Tr}\left[\bM^{-1}\bZ\bZ^\dag\right]\right\}},
\end{equation}
\begin{multline}
f_1(\bZ;\btheta_1) = \left[\frac{1}{\pi^N\det(\bM)}\right]^{(K+1)}\\
\times\exp{\left\{-\mbox{Tr}\left[\bM^{-1}\left((\bz-\alpha\bv)(\bz-\alpha\bv)^\dag+\bS\right)\right]\right\}},
\end{multline}
respectively, where $\btheta_0=\bm{\nu}(\bM)$, $\btheta_1=[\alpha,\btheta_0^T]^T$, {\color{black} and $\bS = K\bS'$ with $\bS' = \frac{1}{K}\sum_{k=1}^K\bz_k\bz_k^\dag$
the sample covariance matrix based upon secondary data}.

Let us now introduce a latent variable
$c$, whose Probability Mass Function (PMF) is $P(c=t)=p_t>0$, $t=0,1$, with $p_0+p_1=1$.
Moreover, we assume that when $c=0$, then $\alpha = 0$, whereas if $c\neq 0$, then $\alpha\neq 0$. It follows that the PDF of $\bZ$ can be recast as
\begin{equation}\label{superposition}
f(\bZ;\deltabf) = p_0f_0(\bZ;\btheta_0)+p_1f_1(\bZ;\btheta_1),
\end{equation}
where $\deltabf = [p_0,p_1,\btheta_0^T,\btheta_1^T]^T$. It turns out that the
PDF of $\bZ$ becomes a convex linear combination of the PDFs under $H_0$ and $H_1$.
In this context, since $c$ is not observable, it is difficult to estimate its PMF
through the maximum likelihood approach and, hence, we estimate the unknown parameters of $f(\bZ;\deltabf)$ by resorting to the
EM algorithm which iteratively provides closed-form estimate
updates towards at least one local stationary point.
More precisely, in the cyclic optimization procedure of the EM algorithm, the $l$th iteration
contains two steps referred to as E-step and M-step \cite{Pia2021Learning,HeteroEM2024anglo}, respectively.
The E-step leads to the computation of the {\em a posteriori} probabilities of each hypothesis given $\bZ$,
namely,
\begin{eqnarray}\label{q0}
q_0^{(l-1)}(\bZ)
=\frac{\widehat{p}_0^{(l-1)}f_0(\bZ;\widehat{\btheta}_0^{(l-1)})}{\widehat{p}^{(l-1)}_0f_0(\bZ;\widehat{\btheta}_0^{(l-1)})
+\widehat{p}_1^{(l-1)} f_{1}(\bZ;\widehat{\btheta}_{1}^{(l-1)})},
\end{eqnarray}
\begin{eqnarray}\label{qm}
q_1^{(l-1)}(\bZ)=\frac{\widehat{p}^{(l-1)}_1 f_1(\bZ;\widehat{\btheta}_{1}^{(l-1)})}
{\widehat{p}^{(l-1)}_0f_0(\bZ;\widehat{\btheta}_0^{(l-1)})
+\widehat{p}^{(l-1)}_{1} f_{1}(\bZ;\widehat{\btheta}_{1}^{(l-1)})},
\end{eqnarray}
with $\widehat{p}_0^{(l-1)}$, $\widehat{\btheta}_0^{(l-1)}$, $\widehat{p}_1^{(l-1)}$ and $\widehat{\btheta}_{1}^{(l-1)}$ being the estimates of $p_0$, $\btheta_0$, $p_1$ and $\btheta_{1}$, at the $(l-1)$th iteration, respectively.
As for the M-step, it consists in solving the following optimization problem
\begin{equation}\label{Mstep}
\widehat{\deltabf}^{(l)}=\arg\max\limits_{\deltabf}\left\{\sum\limits_{t=0}^1 q_t^{(l-1)}(\bZ)\left(\log f_t(\bZ;\btheta_t)+ \log p_t\right)\right\},
\end{equation}
where $\widehat{\deltabf}^{(l)}$ is the estimate of $\deltabf$ at the $l$th iteration. As a result of the iterative application of the E-step and M-step, we obtain a nondecreasing sequence of log-likelihood function values such that $\log{f(\bZ;\deltabf^{(l+1)})}>\log{f(\bZ;\deltabf^{(l)})}$, and this procedure is repeated until a stopping criterion is not satisfied.
The above cyclic procedures need an initial value for $\deltabf$ denoted by $\widehat{\deltabf}^{(0)}$.

{\color{black}As a final remark, notice that the hidden (or missing) random variable $c$ can be viewed
as a random parameter with an {\em a priori} distribution given by $p_t$, $t=0,1$. Thus, problem \eqref{hypotheses} encompasses
both deterministic and random parameters. This point is important to better frame
the detector devised in the next subsection.}
\subsection{Detection Architecture Design}
{\color{black}
From the above procedures, it turns out that the E-steps iteratively update the
estimates of {\em a posteriori} probabilities which can be used to evaluate the ``responsibility''
that the hypotheses take for ``explaining'' the data \cite{bishop2006pattern}.
Specifically, let us consider the Bayesian framework\footnote{Notice that we are implicitly treating
$c$ as a random parameter.}
and if the all the parameters
were known, the decision rule based upon the minimization of the Bayesian risk would be
\be
\frac{p_1 f_1(\bZ;\btheta_1)}{p_0 f_0(\bZ;\btheta_0)}\testp \frac{C_{10}-C_{00}}{C_{01}-C_{11}},
\label{eqn:bayesRisk_detector}
\ee
where $C_{ij}$, $i=0,1$, $j=0,1$, is the cost for selecting $H_i$ when $H_j$ is true.
However, assuming the perfect knowledge of the parameter values is not reasonable in practice and, hence,
we replace them with estimates returned by the EM algorithm to come up with
\be
\frac{\widehat{p}^{(l_{\max})}_1 f_1(\bZ;\widehat{\btheta}^{(l_{\max})}_1)}
{\widehat{p}^{(l_{\max})}_0 f_0(\bZ;\widehat{\btheta}^{(l_{\max})}_0)}\testp \frac{C_{10}-C_{00}}{C_{01}-C_{11}}.
\label{eqn:bayesRisk_detector_MLest}
\ee
Thus, the considered framework jointly exploits the Bayesian and and the ML approach \cite{jointMAP_ML}.
Finally, the costs are used to tune $P_{fa}$, which is of primary concern in radar
\cite{Richards}, by defining the following detection threshold
\be
\eta=\frac{C_{10}-C_{00}}{C_{01}-C_{11}}.
\ee
The corresponding detection architecture becomes
\be
\frac{\widehat{p}^{(l_{\max})}_1 f_1(\bZ;\widehat{\btheta}^{(l_{\max})}_1)}
{\widehat{p}^{(l_{\max})}_0 f_0(\bZ;\widehat{\btheta}^{(l_{\max})}_0)}\testp\eta.
\label{bayes_ML_detector}
\ee
Three remarks are now in order. First, notice that $\eta$ can be set according to the desired $P_{fa}$.
Second, the left-hand side of \eqref{bayes_ML_detector}
shares a similar
structure with the GLRT that, however, does not exploit
any estimate of $p_0$ and $p_1$.
It is also important to highlight that \eqref{bayes_ML_detector}
is statistically equivalent to
\be
\frac{q_1^{(l_{\max})}(\bZ)}{q_0^{(l_{\max})}(\bZ)}\testp\eta,
\ee
that can be viewed as a ``generalized MAP'' detector. The word ``generalized''
is used to indicate that the above detector can be considered a modification of
the MAP rule obtained by introducing $\eta$ and
by replacing the unknown parameters with ML-based estimates.
In the following, this detector
\label{eqn:bayes_ML_detector} will be referred to as Joint Bayesian-ML Detector
based on EM algorithm (EM-BML-D).
}

The final expression of \eqref{bayes_ML_detector} can be obtained by solving \eqref{Mstep}. To this end, let us  observe that in \eqref{Mstep}, the optimization with respect to $\bp=[p_0,p_1]^T$ and $\btheta=[\btheta_0^T,\btheta_{1}^T]^T$ are independent. Thus, the estimates of $\bp$ and $\btheta$ at the $l$th iteration can be separately obtained as
\begin{subequations}\label{separate}
\begin{multline}\label{s-theta}
\widehat{\btheta}^{(l)}=\arg\max\limits_{\bm{\theta}}\left\{q_0^{(l-1)}(\bZ)\log f_0(\bZ;\btheta_0)\right.\\
\left.+q_1^{(l-1)}(\bZ)\log f_{1}(\bZ;\btheta_{1})\right\}
\end{multline}
\begin{equation}\label{s-pi}
\widehat{\bp}^{(l)}=\arg\max\limits_{\bp}\sum\limits_{t=0}^1 q_t^{(l-1)}(\bZ)\log p_t.
\end{equation}
\end{subequations}
It can be shown that the estimates of $p_0$ and $p_1$ are given by
\begin{eqnarray}\label{p}
\widehat{p}_0^{(l)} =q_0^{(l-1)}(\bZ),\nonumber \\
\widehat{p}_1^{(l)} =q_1^{(l-1)}(\bZ).
\end{eqnarray}
Now, the optimization problem with respect to $\btheta$ can be recast as
\begin{multline}\label{theta}
\max\limits_{\text{\bM},\alpha}\left\{q_0^{(l-1)}(\bZ)\log f_0(\bZ;\btheta_0)
+q_1^{(l-1)}(\bZ)\log f_{1}(\bZ;\btheta_{1})\right\},\\
\Rightarrow\max\limits_{\text{\bM},\alpha}\left\{-(K+1)\log{\det{(\bM)}}
-\text{Tr}\left[\bM^{-1}\left(q_0^{(l-1)}(\bZ)\bZ\bZ^\dag\right.\right.\right.\\
\left.\left.\left.\qquad+q_1^{(l-1)}(\bZ)
\left((\bz-\alpha\bv)(\bz-\alpha\bv)^\dag+\bS\right)\right)\right]\right\}.
\end{multline}
Setting to zero the first derivative with respect to $\bM$ yields
\begin{multline}
-(K+1)(\bM^T)^{-1}+q_0^{(l-1)}(\bZ)\left[(\bM^T)^{-1}\left(\bZ\bZ^\dag\right)^T(\bM^T)^{-1}\right]\\
+q_1^{(l-1)}(\bZ)\left[(\bM^T)^{-1}\left((\bz-\alpha\bv)(\bz-\alpha\bv)^\dag+\bS\right)^T(\bM^T)^{-1}\right]\\
 = \bf 0.
\end{multline}
Thus, the estimate of $\bM$ at the $l$th M-step is given by
\begin{multline}\label{Mestimate}
\widehat{\bM}_{(l)} = \frac{1}{K+1}\left[q_0^{(l-1)}(\bZ)\bZ\bZ^\dag\right.\\
\left.+q_1^{(l-1)}(\bZ)
\left((\bz-\alpha\bv)(\bz-\alpha\bv)^\dag+\bS\right)\right].
\end{multline}

Plugging \eqref{Mestimate} into \eqref{theta} leads to
\begin{eqnarray}\label{alpha}
&&\max\limits_{\alpha}\quad -(K+1)\log\det(\widehat{\bM}_{(l)})\nonumber\\
&\Rightarrow&\min\limits_{\alpha}\quad \log\det\left\{q_0^{(l-1)}(\bZ)\bZ\bZ^\dag\right.\nonumber\\
&&\left.+q_1^{(l-1)}(\bZ)
\left[(\bz-\alpha\bv)(\bz-\alpha\bv)^\dag+\bS\right]\right\}\nonumber\\
&\Rightarrow&\min\limits_{\alpha}\quad \log\det\left\{\bA_{(l)}+q_1^{(l-1)}(\bZ)(\bz-\alpha\bv)(\bz-\alpha\bv)^\dag\right\},\nonumber\\
\end{eqnarray}
where $\bA_{(l)} = q_0^{(l-1)}(\bZ)\bZ\bZ^\dag+q_1^{(l-1)}(\bZ)\bS$.
The determinant in \eqref{alpha} can be further written as
\begin{multline}
\det\left\{\bA_{(l)}+q_1^{(l-1)}(\bZ)(\bz-\alpha\bv)(\bz-\alpha\bv)^\dag\right\}\\
=\det(\bA_{(l)})\left[1+q_1^{(l-1)}(\bZ)(\bz-\alpha\bv)^\dag\bA_{(l)}^{-1}(\bz-\alpha\bv)\right].
\end{multline}
Setting to zero the first derivative of the objective function with respect to $\alpha$ leads to
\begin{equation}\label{alphaestimate}
\widehat{\alpha}^{(l)} = \frac{\bv^\dag\bA_{(l)}^{-1}\bz}{\bv^\dag\bA_{(l)}^{-1}\bv}.
\end{equation}

After $l_{\text{max}}$ iterations, plugging the above estimates into \eqref{bayes_ML_detector}, the {\color{black}EM-BML-D} can be finally written as
\begin{multline}\label{testfinal}
\frac{\widehat{p}^{(l_{\text{max}})}_1}{\widehat{p}^{(l_{\text{max}})}_0}
\exp\left\{\bz^\dag\widehat{\bM}_{(l_\text{max})}^{-1}\bz\right.\\
\left.-(\bz-\widehat{\alpha}^{(l_\text{max})}\bv)^\dag
\widehat{\bM}_{(l_\text{max})}^{-1}(\bz-\widehat{\alpha}^{(l_\text{max})}\bv)^\dag\right\} \testinv \eta.
\end{multline}

In next section, we focus on the CFAR property of \eqref{testfinal} with respect to $\bM$, and show that
the test statistic of {\color{black}EM-BML-D} is invariant with respect to the interference covariance matrix under a specific initialization.

\section{CFAR properties}

\label{CFAR}
Since the decision statistic in \eqref{bayes_ML_detector} depends on the iteration index $l$, we resort
to the mathematical induction to prove the following proposition.

\begin{proposition}
$\frac{q_1^{(l)}(\bZ)}{q_0^{(l)}(\bZ)}$ is CFAR with respect to $\bM$ provided that the distribution
of $\frac{q_1^{(l-1)}(\bZ)}{q_0^{(l-1)}(\bZ)}$ is invariant with respect to $\bM$.
\end{proposition}

\begin{proof}
Let us start by observing that from  \eqref{p}, when the distribution of  $\frac{q_1^{(l-1)}(\bZ)}{q_0^{(l-1)}(\bZ)}$
does not depend on $\bM$, then also that of $\frac{\widehat{p}^{(l)}_1}{\widehat{p}^{(l)}_0}$ exhibits this property.
Furthermore, notice that\footnote{Notice that we do not consider the case that $\widehat{p}_0$,
$\widehat{p}_1$, $q_0$, or $q_1=0$.}
\begin{eqnarray}\label{q0CFAR}
q_0^{(l-1)}(\bZ)
&=&\frac{\widehat{p}_0^{(l-1)}f_0(\bZ;\widehat{\btheta}_0^{(l-1)})}{\widehat{p}^{(l-1)}_0f_0(\bZ;\widehat{\btheta}_0^{(l-1)})
+\widehat{p}_1^{(l-1)} f_{1}(\bZ;\widehat{\btheta}_{1}^{(l-1)})}\nonumber\\
&=&\frac{1}{1+\frac{\widehat{p}_1^{(l-1)}f_1(\bZ;\widehat{\btheta}_1^{(l-1)})}
{\widehat{p}_0^{(l-1)}f_0(\bZ;\widehat{\btheta}_0^{(l-1)})}}\nonumber\\
&=&\frac{1}{1+\frac{q_1^{(l-1)}(\bZ)}{q_0^{(l-1)}(\bZ)}}.
\end{eqnarray}
Therefore, the distribution of
$q_0^{(l-1)}(\bZ)$ is independent of $\bM$. The same line of reasoning can be applied to show that
$q_1^{(l-1)}(\bZ)$ has a distribution independent of $\bM$.

From \eqref{testfinal}, let us define
\begin{equation}
g(\widehat{\btheta}^{(l)}_0,\widehat{\btheta}^{(l)}_1;\bZ) =
\bz^\dag\widehat{\bM}_{(l)}^{-1}\bz
-(\bz-\widehat{\alpha}^{(l)}\bv)^\dag
\widehat{\bM}_{(l)}^{-1}(\bz-\widehat{\alpha}^{(l)}\bv)^\dag,
\end{equation}
which can be written as
\begin{eqnarray}
&&g(\widehat{\btheta}^{(l)}_0,\widehat{\btheta}^{(l)}_1;\bZ)\nonumber\\
&=&\bz^\dag\widehat{\bM}_{(l)}^{-1}\bz-\bz^\dag\widehat{\bM}_{(l)}^{-1}\bz+
(\widehat{\alpha}^{(l)})^\dag\bv^\dag\widehat{\bM}_{(l)}^{-1}\bz\nonumber\\
&&+\widehat{\alpha}^{(l)}\bz^\dag\widehat{\bM}_{(l)}^{-1}\bv-|\widehat{\alpha}^{(l)}|^2
\bv^\dag\widehat{\bM}_{(l)}^{-1}\bv\nonumber\\
&=&\frac{\bz^\dag\bA_{(l)}^{-1}\bv\bv^\dag\widehat{\bM}_{(l)}^{-1}\bz}{\bv^\dag\bA_{(l)}^{-1}\bv}
+\frac{\bv^\dag\bA_{(l)}^{-1}\bz\bz^\dag\widehat{\bM}_{(l)}^{-1}\bv}{\bv^\dag\bA_{(l)}^{-1}\bv}\nonumber\\
&&-\frac{\bv^\dag\bA_{(l)}^{-1}\bz\bz^\dag\bA_{(l)}^{-1}\bv\bv^\dag\widehat{\bM}_{(l)}^{-1}\bv}
{\left(\bv^\dag\bA_{(l)}^{-1}\bv\right)^2}.
\end{eqnarray}
Let $\bar{\bv} = \bM^{-\frac{1}{2}}\bv$ and $\bar{\bz} = \bM^{-\frac{1}{2}}\bz$,
then we have
\begin{eqnarray}
&&g(\widehat{\btheta}^{(l)}_0,\widehat{\btheta}^{(l)}_1;\bZ)\nonumber\\
&=&\frac{\bar{\bz}^\dag\bM^{\frac{1}{2}}\bA_{(l)}^{-1}\bM^{\frac{1}{2}}\bar{\bv}
\bar{\bv}^\dag\bM^{\frac{1}{2}}\widehat{\bM}_{(l)}^{-1}\bM^{\frac{1}{2}}\bar{\bz}}
{\bar{\bv}^\dag\bM^{\frac{1}{2}}\bA_{(l)}^{-1}\bM^{\frac{1}{2}}\bar{\bv}}\nonumber\\
&&+\frac{\bar{\bv}^\dag\bM^{\frac{1}{2}}\bA_{(l)}^{-1}\bM^{\frac{1}{2}}\bar{\bz}
\bar{\bz}^\dag\bM^{\frac{1}{2}}\widehat{\bM}_{(l)}^{-1}\bM^{\frac{1}{2}}\bar{\bv}}
{\bar{\bv}^\dag\bM^{\frac{1}{2}}\bA_{(l)}^{-1}\bM^{\frac{1}{2}}\bar{\bv}}\nonumber\\
&&-\frac{\bar{\bv}^\dag\bM^{\frac{1}{2}}\bA_{(l)}^{-1}\bM^{\frac{1}{2}}\bar{\bz}
\bar{\bz}^\dag\bM^{\frac{1}{2}}\bA_{(l)}^{-1}\bM^{\frac{1}{2}}\bar{\bv}
\bar{\bv}^\dag\bM^{\frac{1}{2}}\widehat{\bM}_{(l)}^{-1}\bM^{\frac{1}{2}}\bar{\bv}}
{\left(\bar{\bv}^\dag\bM^{\frac{1}{2}}\bA_{(l)}^{-1}\bM^{\frac{1}{2}}\bar{\bv}\right)^2}\nonumber\\
&=&\frac{\bar{\bz}^\dag\bar{\bA}_{(l)}^{-1}\bar{\bv}\bar{\bv}^\dag\bar{\bM}_{(l)}^{-1}\bar{\bz}}
{\bar{\bv}^\dag\bar{\bA}_{(l)}^{-1}\bar{\bv}}
+\frac{\bar{\bv}^\dag\bar{\bA}_{(l)}^{-1}\bar{\bz}\bar{\bz}^\dag\bar{\bM}_{(l)}^{-1}\bar{\bv}}
{\bar{\bv}^\dag\bar{\bA}_{(l)}^{-1}\bar{\bv}}\nonumber\\
&&-\frac{\bar{\bv}^\dag\bar{\bA}_{(l)}^{-1}\bar{\bz}\bar{\bz}^\dag\bar{\bA}_{(l)}^{-1}\bar{\bv}\bar{\bv}^\dag\bar{\bM}_{(l)}^{-1}\bar{\bv}}
{\left(\bar{\bv}^\dag\bar{\bA}_{(l)}^{-1}\bar{\bv}\right)^2},
\end{eqnarray}
where
\begin{eqnarray}
\bar{\bA}_{(l)} &=& \bM^{-\frac{1}{2}}\bA_{(l)}\bM^{-\frac{1}{2}}\nonumber\\
&=&q_0^{(l-1)}(\bZ)\bM^{-\frac{1}{2}}\bZ\bZ^\dag\bM^{-\frac{1}{2}}+q_1^{(l-1)}(\bZ)\bM^{-\frac{1}{2}}\bS\bM^{-\frac{1}{2}}
\nonumber\\
\end{eqnarray}
\begin{eqnarray}
\bar{\bM}_{(l)}
&=& \bM^{-\frac{1}{2}}\widehat{\bM}_{(l)}\bM^{-\frac{1}{2}}\nonumber\\
&=&\frac{1}{K+1}\left[q_0^{(l-1)}(\bZ)\bM^{-\frac{1}{2}}\bZ\bZ^\dag\bM^{-\frac{1}{2}}+\right.\nonumber\\
&&q_1^{(l-1)}(\bZ)
\left(\bM^{-\frac{1}{2}}(\bz-\widehat{\alpha}^{(l)}\bv)(\bz-\widehat{\alpha}^{(l)}\bv)^\dag\bM^{-\frac{1}{2}}+
\right.\nonumber\\
&&\left.\left.\bM^{-\frac{1}{2}}\bS\bM^{-\frac{1}{2}}\right)\right]\nonumber\\
&=&\frac{1}{K+1}\left[\bar{\bA}_{(l)}+q_1^{(l-1)}(\bZ)(\bar{\bz}-\widehat{\alpha}^{(l)}\bar{\bv})
(\bar{\bz}-\widehat{\alpha}^{(l)}\bar{\bv})^\dag\right].\nonumber\\
\end{eqnarray}
Since $\bM^{-\frac{1}{2}}\bZ\bZ^\dag\bM^{-\frac{1}{2}} \sim\cC\cW(K+1,N;\bI_N)$ and
$\bM^{-\frac{1}{2}}\bS\bM^{-\frac{1}{2}} \sim\cC\cW(K,N;\bI_N)$, it follows
that $\bar{\bA}_{(l)}$ is invariant to the interference covariance matrix $\bM$.

Now, we introduce a unitary transformation that rotates the
whitened steering vector into the first elementary vector
\begin{equation}
\bU^\dag\bar{\bv} = \sqrt{\bv^\dag\bM^{-1}\bv}\boe_1,
\end{equation}
where $\bU\in\mathds{C}^{N\times N}$ is a unitary matrix
and $\boe_1=[1,0,\ldots,0]^T\in\mathds{C}^{N\times 1}$.
Then $g(\widehat{\btheta}^{(l)}_0,\widehat{\btheta}^{(l)}_1;\bZ)$ can be recast as
\begin{eqnarray}
g(\widehat{\btheta}^{(l)}_0,\widehat{\btheta}^{(l)}_1;\bZ)
=&\frac{\bar{\bar{\bz}}^\dag\bar{\bar{\bA}}_{(l)}^{-1}\boe_1\boe_1^\dag\bar{\bM}_{(l)}^{-1}\bar{\bar{\bz}}}
{\boe_1^\dag\bar{\bA}_{(l)}^{-1}\boe_1}
+\frac{\boe_1^\dag\bar{\bar{\bA}}_{(l)}^{-1}\bar{\bar{\bz}}\bar{\bar{\bz}}^\dag\bar{\bar{\bM}}_{(l)}^{-1}\boe_1}
{\boe_1^\dag\bar{\bar{\bA}}_{(l)}^{-1}\boe_1}\nonumber\\
&-\frac{\boe_1^\dag\bar{\bar{\bA}}_{(l)}^{-1}\bar{\bar{\bz}}\bar{\bar{\bz}}^\dag\bar{\bar{\bA}}_{(l)}^{-1}\boe_1\boe_1^\dag\bar{\bar{\bM}}_{(l)}^{-1}\boe_1}
{\left(\boe_1^\dag\bar{\bar{\bA}}_{(l)}^{-1}\boe_1\right)^2},
\end{eqnarray}
where
\begin{equation}
\left\{
\begin{array}{lll}
\bar{\bar{\bz}} &=& \bU^\dag\bar{\bz}\sim\cC\cN(0,\bI_N),\\
\bar{\bar{\bA}}_{(l)} &=& \bU^\dag\bar{\bA}_{(l)}\bU \ \text{obeys the same distribution as} \ \bar{\bA}_{(l)},\\
\bar{\bar{\bM}}_{(l)} &=& \bU^\dag\bar{\bM}_{(l)}\bU\\
&\propto&\left[\bar{\bar{\bA}}_{(l)}+q_1^{(l-1)}(\bZ)(\bar{\bar{\bz}}-\widehat{\alpha}_w^{(l)}\boe_1)
(\bar{\bar{\bz}}-\widehat{\alpha}_w^{(l)}\boe_1)^\dag\right]
\end{array}
\right.
\end{equation}
with $\widehat{\alpha}_w^{(l)} = \sqrt{\bv^\dag\bM^{-1}\bv}\widehat{\alpha}^{(l)} =
\frac{\boe_1^\dag\bar{\bar{\bA}}_{(l)}^{-1}\bar{\bar{\bz}}}{\boe_1^\dag\bar{\bar{\bA}}_{(l)}^{-1}\boe_1}$.

With all of the results above, it can be seen that when $\frac{q_1^{(l-1)}(\bZ)}{q_0^{(l-1)}(\bZ)}$
is CFAR with respect to $\bM$, $\frac{q_1^{(l)}(\bZ)}{q_0^{(l)}(\bZ)}$
includes no real interference covariance $\bM$ in the statistical characteristic, and thus ensures the CFAR property.
\end{proof}

Gathering the above results, it turns out that when
$\frac{q_1^{(0)}(\bZ)}{q_0^{(0)}(\bZ)}$ is independent of $\bM$, we can claim that for a given $l_{\text{max}}$,
the {\color{black}EM-BML-D} proposed in \eqref{testfinal} possesses the CFAR with respect to $\bM$.

\section{Performance Assessment}

In this section, we assess the performance of the {\color{black}EM-BML-D}
in comparison with well-known competitors such as Kelly's GLRT, the AMF, Rao, and the ACE \cite{Kraut-PHE,AsympComGau1995}.
{\color{black}At the same time, in order to learn the performance degradation of the proposed algorithm due to parameter estimation,
we consider the test given by \eqref{bayes_ML_detector} with all parameters known (except for $p_0$ and $p_1$) as a benchmark.
This detector is referred to in the ensuing numerical examples as ``Benchmark''.}
To this end, we resort to
the Monte Carlo counting techniques and obtain the detection threshold corresponding to a pre-assigned
$P_{fa}$ through $100/P_{fa}$ independent trials, and the Probability of Detection ($P_d$) through 1000
independent trials, respectively.

Specifically, we investigate the following aspects:
\begin{enumerate}
  \item the convergence rate of the cyclic procedure based upon the EM algorithm;
  \item the CFAR behavior of the {\color{black}EM-BML-D} with respect to the disturbance covariance matrix;
  \item the detection performance in the presence of perfectly matched signals under different configurations;
  \item the rejection capability of mismatched target echoes.
\end{enumerate}

\subsection{Initialization Issues and Simulation Setup}
\label{Initia}
In the next illustrative examples, we assume that the interference covariance matrix is
$\bM=\sigma^2\bI_N+\sigma_c^2\bM_c$, where $\sigma^2\bI_N$ is the thermal noise component
with $\sigma^2=1$ the noise power
and $\sigma_c^2>0$ is the clutter power set according to the Clutter-to-Noise-Ratio (CNR)
given by $\text{CNR} = 10\log(\sigma_c^2/\sigma^2)$. As for the clutter covariance matrix $\bM_c$,
we assume that its $(i,j)$th entry is equal to $\rho^{|i-j|}$ with $\rho=0.9$ the
one-lag correlation coefficient.
Moreover, we define the Signal-to-Clutter-plus-Noise-Ratio (SCNR) as
\begin{equation}\label{SCNR}
\text{SCNR} = |\alpha|^2\bv_t^\dag\bM^{-1}\bv_t,
\end{equation}
where $\bv_t$ is the true target steering vector that is equal to the nominal steering vector $\bv$ in a perfectly matched
scenario.

For initialization purpose,
we set $\widehat{p}_0^{(0)}=\widehat{p}_1^{(0)}=1/2$,
while $\widehat{\bM}_{(0)}$ is equal to $\bS$. As for the initial value of $\widehat{\alpha}^{(0)}$, a
reasonable choice is
\begin{equation}
\widehat{\alpha}^{(0)} = \frac{\bv^\dag\bS^{-1}\bz}{\bv^\dag\bS^{-1}\bv}.
\end{equation}
It is important to underline that with this initializations, $\frac{q_1^{(0)}(\bZ)}{q_0^{(0)}(\bZ)}$
can be written as
\begin{eqnarray}
\frac{q_1^{(0)}(\bZ)}{q_0^{(0)}(\bZ)}&=&\frac{\widehat{p}_0^{(0)}}{\widehat{p}_1^{(0)}}
\exp\left\{g(\widehat{\alpha}^{(0)},\bS;\bZ)\right\}\nonumber\\
&=&\frac{\left|\bv^\dag\bS^{-1}\bz\right|^2}{\bv^\dag\bS^{-1}\bv}
=\frac{\left|\boe_1^\dag\bar{\bar{\bS}}^{-1}\bar{\bar{\bz}}\right|^2}{\boe_1^\dag\bar{\bar{\bS}}^{-1}\boe_1},
\end{eqnarray}
where $\bar{\bar{\bS}} = \bU^\dag\bM^{-\frac{1}{2}}\bS\bM^{-\frac{1}{2}}\bU\sim\cC\cW(K,N;\bI_N)$ indicates
that $\frac{q_1^{(0)}(\bZ)}{q_0^{(0)}(\bZ)}$ is CFAR with respect to $\bM$.

\subsection{The Convergence of the Cyclic Procedure}

In this subsection, we investigate the convergence rate of the iterative estimation procedure.
To this end, let $\cL(\bZ;\btheta)$ be the objective
function in \eqref{theta}, then the {\color{black} relative variation of the likelihood function} due to the estimation of $\btheta$ at the
$l$th iteration is defined as
\begin{equation}
\Delta\cL(l)=\left|\frac{\cL(\bZ;\widehat{\btheta}^{(l)})-\cL(\bZ;\widehat{\btheta}^{(l-1)})}
{\cL(\bZ;\widehat{\btheta}^{(l)})}\right|.
\end{equation}

Fig. \ref{fig-IterationNumber-H0} shows $\Delta\cL(l)$ versus $l$ under $H_0$ averaged over 1000 independent
Monte Carlo trials for different configurations of $N$ and $K$.
As it is observed, the curves rapidly decrease as $l$ increases and
for high values of $N$ and $K$, the variation further decreases due to the increased estimation quality.
More importantly, for the considered cases under $H_0$,
the average relative variation is lower than $10^{-4}$ for $l\geq4$ and reduces to $10^{-5}$ at $l=6$.
Such a behavior are confirmed in Fig. \ref{fig-IterationNumber-H1} which corresponds to
$H_1$ and contains the curves of $\Delta\cL(l)$ for different SCNR values.
Specifically, the variation is lower than $10^{-5}$ when $l\geq 5$.
Finally, notice that for a given configuration,
the variation significantly reduces when the SCNR increases.

Summarizing,
as a compromise between the computational burden and the optimization performance,
we use $l_{\text{max}}=5$ in the ensuing simulations. For the sake of comparison, we also consider
$l_{\text{max}}=7$.

\begin{figure}[h]
  \centering
  \includegraphics[width=7.7cm]{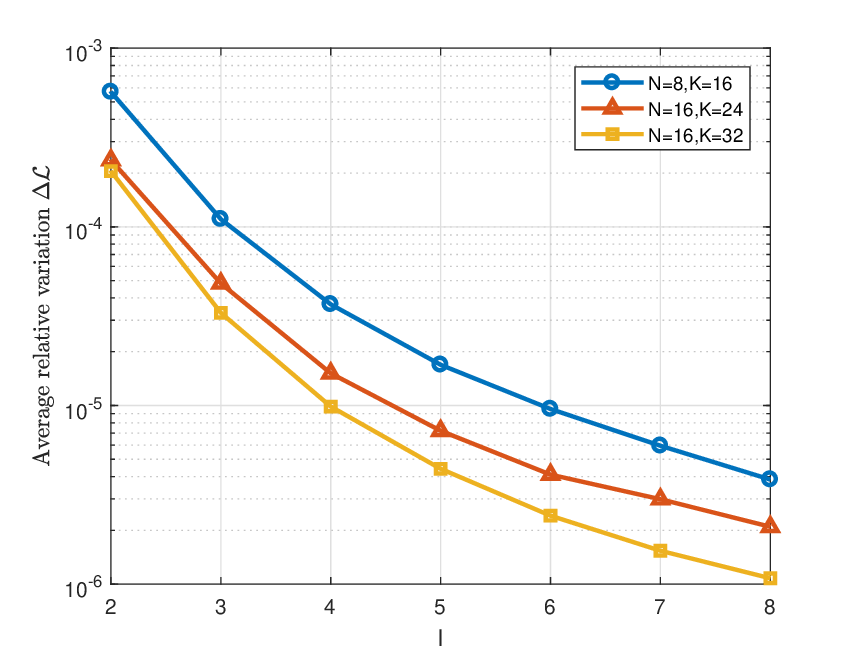}\\
  \caption{$\Delta\cL(l)$ versus iteration number under $H_0$, $\rho=0.9$, CNR$=30$ dB.}\label{fig-IterationNumber-H0}
\end{figure}

\begin{figure}[h]
  \centering
  \includegraphics[width=7.7cm]{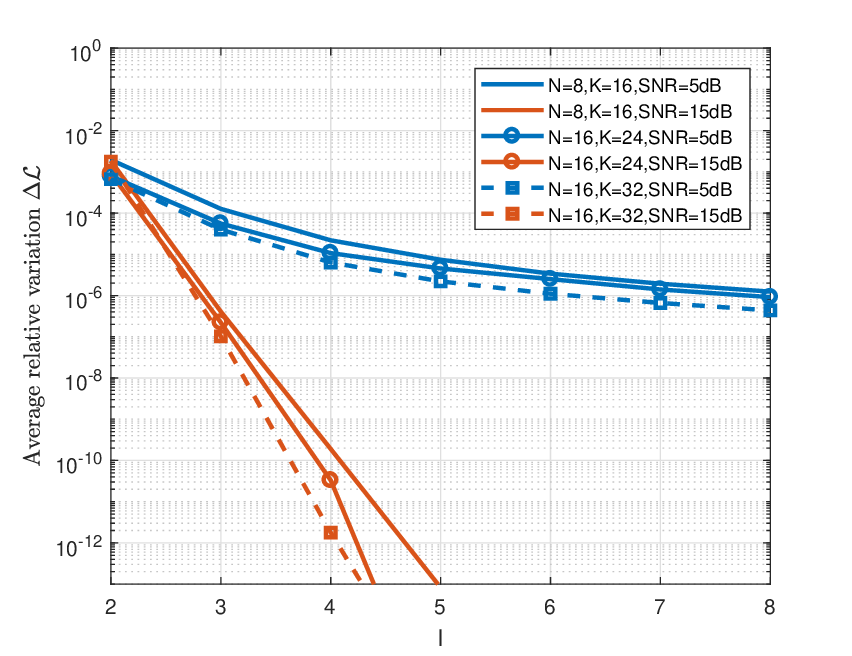}\\
  \caption{$\Delta\cL(l)$ versus iteration number under $H_1$, $\rho=0.9$, CNR$=30$ dB.}\label{fig-IterationNumber-H1}
\end{figure}
\subsection{The CFAR Properties}

Let us recall that in Section \ref{CFAR}, we proved the CFAR property of the {\color{black}EM-BML-D} with respect to
the statistical properties of the disturbance. In this subsection, we corroborate the results
presented in Section \ref{CFAR} by resorting to Monte Carlo simulation. For completeness,
we also consider the aforementioned competitors.
In, Fig. \ref{fig-CFAR-sigma} and Fig. \ref{fig-CFAR-rho}, we use thresholds corresponding to a nominal
$P_{fa}=10^{-4}$ with $N=8$, $K=16$, $\rho=0.9$, and CNR$=30$ dB to evaluate the actual $P_{fa}$
under different CNR and $\rho$ values. In the figures,
{\color{black}EM-BML-D}5 and {\color{black}EM-BML-D}7 stand for the {\color{black}EM-BML-D} with $l_{\text{max}}=5$ and $l_{\text{max}}=7$, respectively.
It can be seen from the figures that the {\color{black}EM-BML-D} is capable of maintaining the pre-assigned $P_{fa}$
when the CNR varies from
30 dB to 110 dB and $\rho$ changes from 0.5 to 0.9.

Such a corroboration of the CFAR property can also be observed under other configurations of
$N$ and $K$, and thus omitted here for brevity.

\begin{figure}[h]
  \centering
  \includegraphics[width=7.7cm]{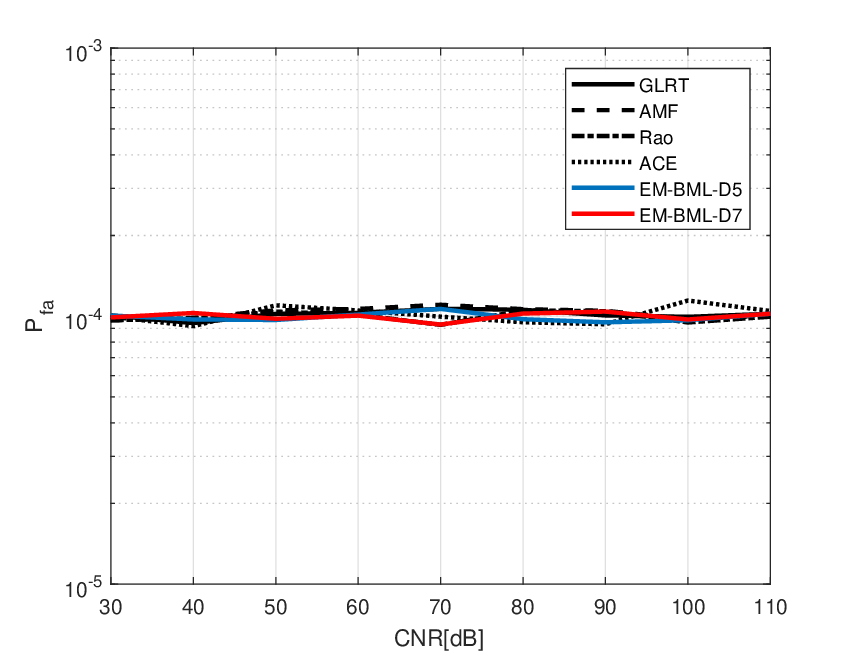}\\
  \caption{$P_{fa}$ versus CNR, assuming $N=8$, $K=16$, $\rho=0.9$, nominal CNR$=30$ dB.}\label{fig-CFAR-sigma}
\end{figure}

\begin{figure}[h]
  \centering
  \includegraphics[width=7.7cm]{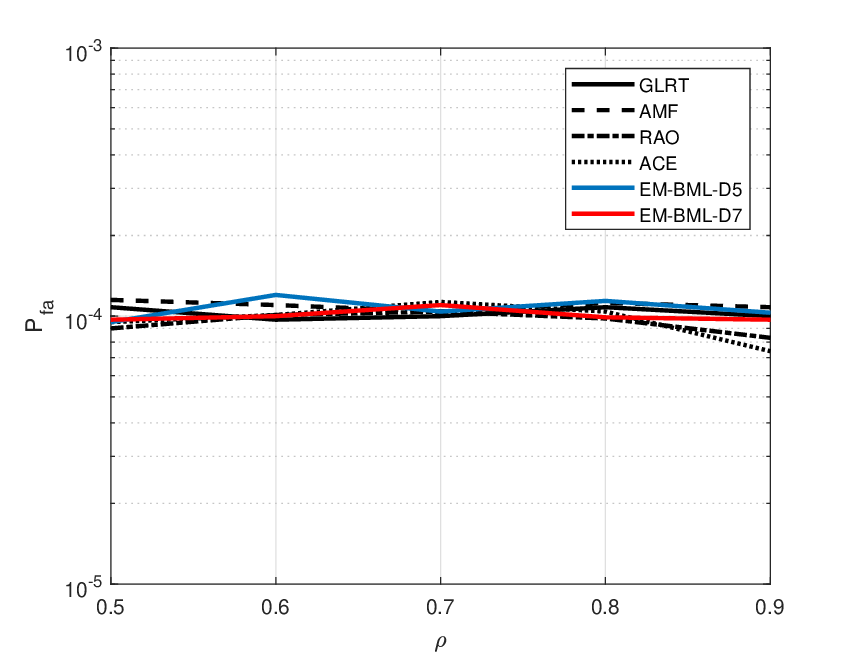}\\
  \caption{$P_{fa}$ versus $\rho$, assuming $N=8$, $K=16$, CNR$=30$ dB, nominal $\rho=0.9$.}\label{fig-CFAR-rho}
\end{figure}
\subsection{Matched Targets}

The detection performance analysis starts from the case where the nominal
steering vector $\bv$ and the actual steering vector $\bv_t$ are perfectly aligned
or at least the misalignment is negligible. Moreover, we initially set $K=2N$, which allows us to limit the detection loss related to the estimation of the interference parameters as claimed in \cite{RMB},
then we consider the so-called {\em sample-starved} scenario \cite{Hao2014Pers} where $N<K<2N$.

Fig. \ref{fig-Pd-N8K16} shows the $P_d$ curves assuming $N=8$, $K=16$, and highlights
that at least for the considered operating scenario, the proposed MAP-based detectors significantly outperform their conventional counterparts.
To be more specific, it can be seen that the {\color{black}EM-BML-D}5 and {\color{black}EM-BML-D}7 ensure the best performance with $P_d>0.8$ when
$\text{SCNR}=15$ dB, whereas the GLRT and AMF experience a $P_d$ around 0.6 and 0.5, respectively. The performance enhancement
is due to the fact that the information of both hypotheses has been combined through the latent variables to
come up with the model of the received data, as shown in \eqref{superposition}. Such a combination allows us
to account for the hypothesis priors in the decision and estimation process.
As a matter of fact, the estimates of the unknowns are linear combinations of $H_0$ and $H_1$, with {\em a posteriori} probabilities as the coefficients (see \eqref{Mestimate} and \eqref{alphaestimate}).
At the same time, the EM algorithm endows the detection system with the capability of estimating the posterior information from
data.
Notice also that in Fig. \ref{fig-Pd-N8K16}, the {\color{black}EM-BML-D} with $l_{\text{max}}=5$ and $l_{\text{max}}=7$ share
almost the same performance, which may be due to the sufficiently small residual error of the
cyclic procedure.
{\color{black}In Table \ref{table:time}, we report the execution time (averaged over 1000 Monte Carlo trials) for each of the considered detectors
to provide an estimate of the respective computational complexity. Inspection of Table \ref{table:time} points out that the EM-BML-D achieves better detection performance at
the price of an increased computational burden (about one order of magnitude with respect to the competitors) due to the iterative estimation procedures. As expected, the EM-BML-D7 is more time-consuming
than the EM-BML-D5. Similar results have been obtained under different parameter configurations and are not reported here for brevity.
\begin{table}[h]
\caption{Average Execution Times over 1000 Monte Carlo Trails}
\label{table:time}
\centering
\begin{tabular}{c|ccccc}
  \toprule
   Detector&GLRT & AMF & Rao \\
  \midrule
  Time(s)&$2.29\times10^{-5}$&$1.93\times10^{-5}$&$2.30\times10^{-5}$\\
  \bottomrule
  \toprule
  Detector&ACE & {\color{black}EM-BML-D5}&EM-BML-D7  \\
  \midrule
  Time(s)&$2.21\times10^{-5}$&$3.81\times10^{-4}$ &$5.37\times10^{-4}$\\
   \bottomrule
\end{tabular}
\end{table}
}

To further corroborate the results observed in Fig. \ref{fig-Pd-N8K16}, in Fig. \ref{fig-Pd-N16} we show
the detection performance of the detectors assuming a parameter setting equal to that of Fig. \ref{fig-Pd-N8K16}, except for
the values of $N$ and $K$. It can be seen that the {\color{black}EM-BML-D}5 shares the same performance as {\color{black}EM-BML-D}7 and ensures the highest performance
in both figures. Notice that when $N=16$, $K=32$, the improvement of the {\color{black}EM-BML-D} with respect to the counterparts decreases since the quality of covariance matrix estimates of the latter improves.
From Fig. \ref{fig-Pd-N16K24}, it turns out that even though $N$ and $K$
become larger compared to Fig. \ref{fig-Pd-N8K16}, the detection performance may deteriorate due to the decrease of the ratio $K/N$.
Finally, Fig. \ref{fig-Pd-N16Kvari} shows the $P_d$ curves of {\color{black}EM-BML-D}5 and GLRT when $N=16$ and $K$ as a parameter.
It can be observed that {\color{black}EM-BML-D}5 with $K=28$ outperforms the GLRT with $K=32$.
\begin{figure}[h]
  \centering
  \includegraphics[width=7.7cm]{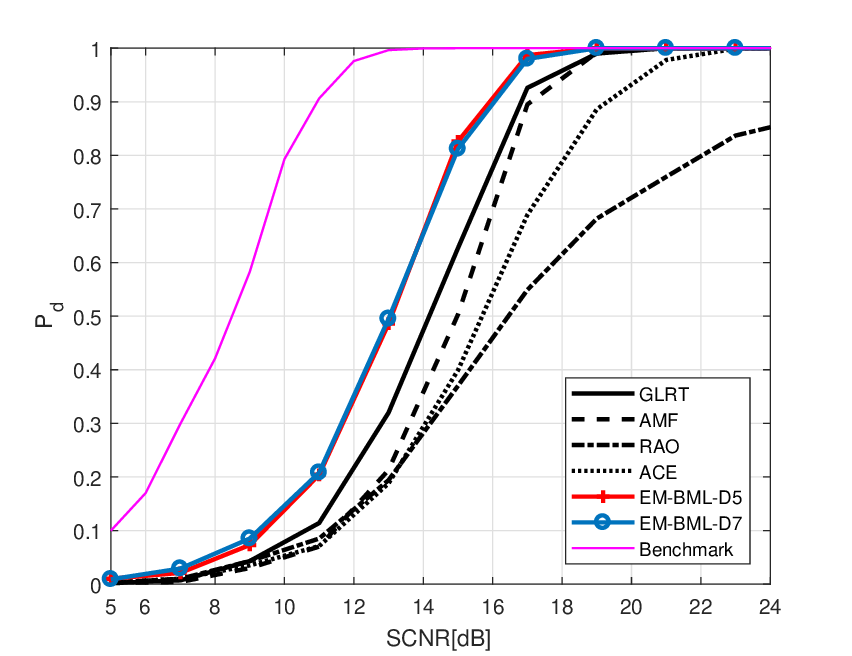}\\
  \caption{$P_d$ versus SCNR for matched targets, $N=8$, $K=16$, $\text{CNR}=30$ dB.}\label{fig-Pd-N8K16}
\end{figure}

\begin{figure}[h]
  \centering
  \subfigure[$N=16$, $K=32$]{\includegraphics[width=7.7cm]{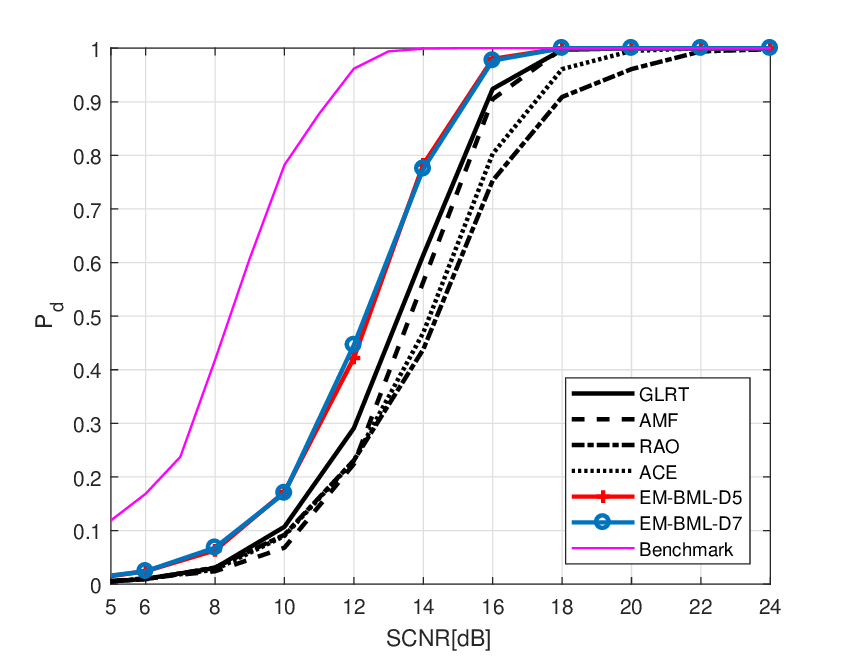}\label{fig-Pd-N16K32}}
  \subfigure[$N=16$, $K=24$]{\includegraphics[width=7.7cm]{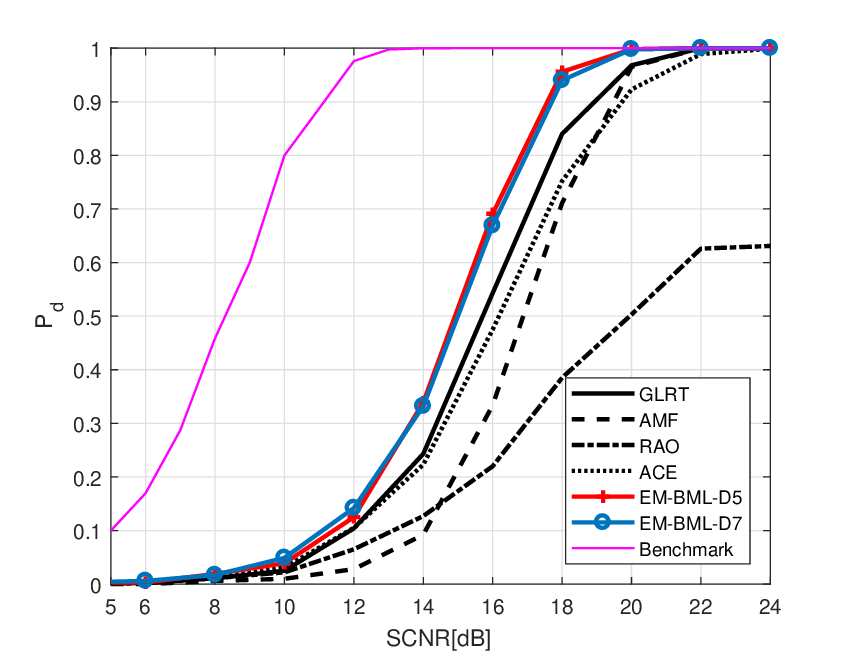}\label{fig-Pd-N16K24}}
  \caption{$P_d$ versus SCNR for matched targets, $\text{CNR}=30$ dB.}\label{fig-Pd-N16}
\end{figure}

\begin{figure}[h]
  \centering
  \includegraphics[width=7.7cm]{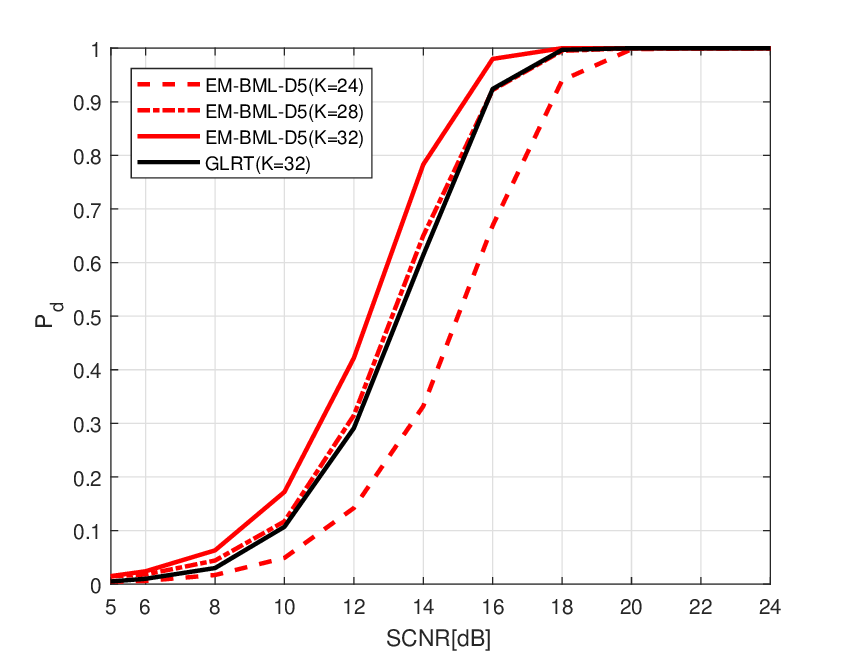}\\
  \caption{$P_d$ versus SCNR for matched targets, $N=16$, $\text{CNR}=30$ dB.}\label{fig-Pd-N16Kvari}
\end{figure}

To conclude this subsection, we assess the detection performance of the {\color{black}EM-BML-D5}
using the real radar data collected by the MIT Lincoln Laboratory Phase One radar in 1985.
The considered data set is {\em H067037.3} which contains 30 720 temporal returns and 76 range bins
of L-band land clutter data with HH polarization. More details about the Phase One radar
can be found in \cite{HaoSpringer} and \cite{Billingsley1999groundclutterdata}. Analyses in \cite{Pia2021Learning}
and \cite{HeteroEM2024anglo} show that such data does not match the underlying assumption upon which
the {\color{black}EM-BML-D} is derived, namely homogeneous environment,
{\color{black}
and thus allows for an evaluation of the capability of adapting to different radar operating scenarios.
Assuming $N=8$ and $K=16$,  the threshold is first computed by considering range bin 28 as the CUT, and $K/2$ bins
on each side of the CUT as secondary data, respectively, and we use 5 pulses of overlap in consecutive trials to
achieve a sufficient number of Monte Carlo trials, namely $100/P_{fa}$ with $P_{fa}=10^{-2}$. Before evaluating $P_d$,
we estimate the $P_{fa}$ of such a threshold on range bin 28 and 66, respectively, as shown in Table \ref{table:28} and \ref{table:49}.
It can be seen that the threshold obtained on range bin 28 can maintain the desired $P_{fa}$ also on bin 66 which is
actually located in an nonhomogeneous region (see \cite{Pia2021Learning}).
}
\begin{table}[h]
\caption{Estimated $P_{fa}$ for range bin 28 ($\times 10^{-2}$)}
\label{table:28}
\centering
\begin{tabular}{ccccc}
  \toprule
   GLRT & AMF & Rao & ACE & {\color{black}EM-BML-D5}  \\
  \midrule
   1.01 &1.00&0.99&0.99&1.00\\
  \bottomrule
\end{tabular}
\end{table}

\begin{table}[h]
\caption{Estimated $P_{fa}$ for range bin 66 ($\times 10^{-2}$)}
\label{table:49}
\centering
\begin{tabular}{ccccc}
  \toprule
   GLRT & AMF & Rao & ACE & {\color{black}EM-BML-D5}  \\
  \midrule
   1.19 &1.7&1.3&1.3&1.23\\
  \bottomrule
\end{tabular}
\end{table}

As for the targets, they are
injected into the CUT as done for synthetic data according to the SCNR defined in \eqref{SCNR}.
Fig. \ref{fig-Pd-realdata-28} reports $P_d$ of the detectors versus SCNR
for range bin 28. It turns out that the {\color{black}EM-BML-D5} maintains a superior performance with respect to
the counterparts confirming
what have been observed in the case of synthetic data.Finally, we plot the curves of $P_d$ versus SCNR using the same thresholds as in Fig. \ref{fig-Pd-realdata-28} with
range bin 66 as the primary data, bins 57-64 and 68-75 as secondary data, in Fig. \ref{fig-Pd-realdata-49}.
Joint inspection of Table \ref{table:49} and Fig. \ref{fig-Pd-realdata-49} indicates that
the {\color{black}EM-BML-D5} enjoys a good capability of maintaining the $P_{fa}$ constant in different regions
and can overcome the competitors in terms of $P_d$.

\begin{figure}[h]
  \centering
  \includegraphics[width=7.7cm]{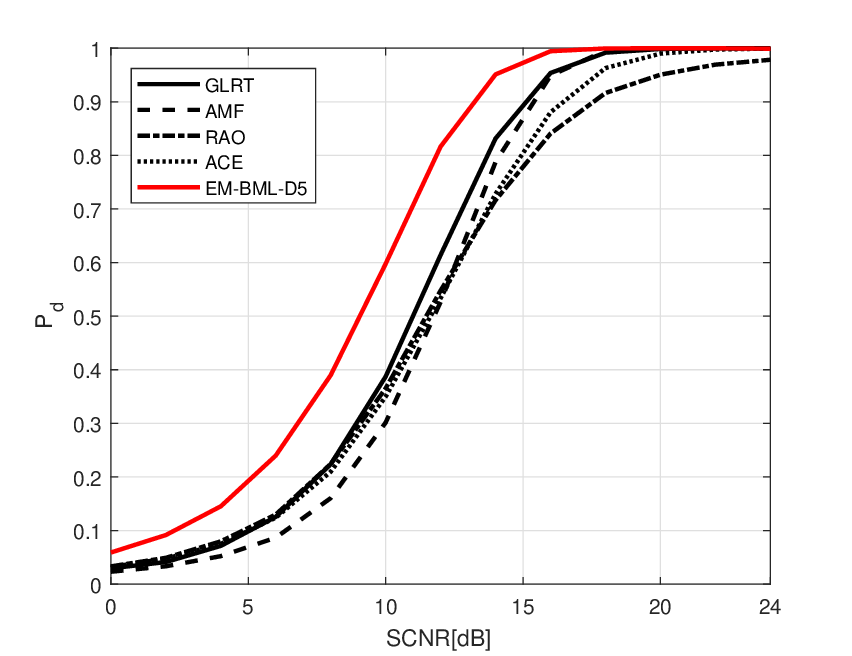}\\
  \caption{$P_d$ versus SCNR for real recorded data, range bin 28 used for both threshold estimation and $P_d$ evaluation, $l_{\text{max}}=5$.}\label{fig-Pd-realdata-28}
\end{figure}

\begin{figure}[h]
  \centering
  \includegraphics[width=7.7cm]{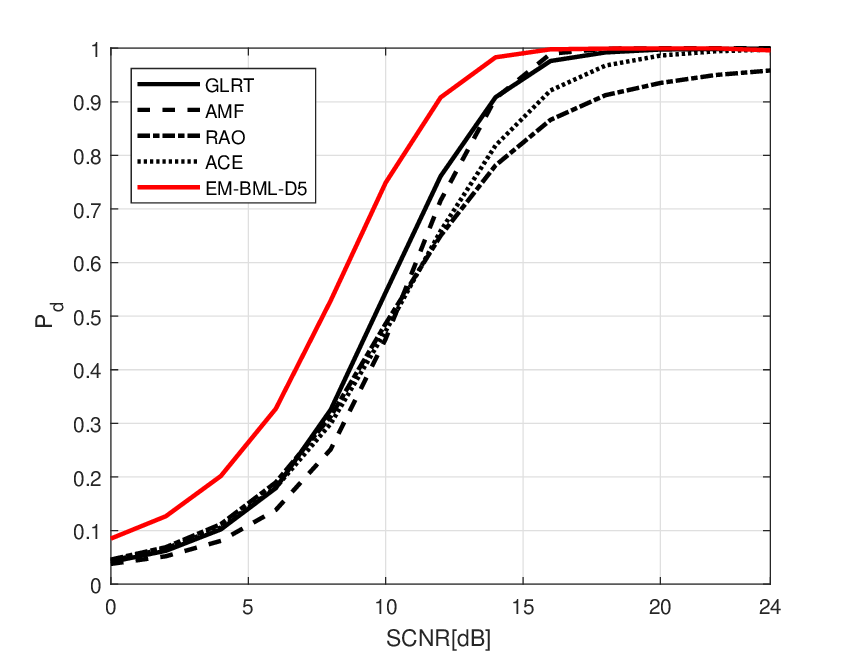}\\
  \caption{$P_d$ versus SCNR for real recorded data, range bin 28 used for threshold estimation, range bin 66 used for $P_d$ evaluation, $l_{\text{max}}=5$.}\label{fig-Pd-realdata-49}
\end{figure}
\subsection{Mismatched Targets}

The analysis presented in this subsection considers the case where
a mismatch between the true target steering vector $\bv_t$ and the nominal one exists.
This analysis allows us to understand if the proposed detector is selective or
robust \cite{Bandiera2009}. Let us recall here that a selective detector rejects slightly
mismatched signals with high probability whereas a robust detector is capable of
guaranteeing high detection probabilities in the presence of uncertainties
in the nominal steering vector.
To this end, we
evaluate the rejection capability of the detectors in terms of the mismatch angle $\phi$ defined as
\begin{equation}\label{costheta}
\cos^2\phi =\frac{\left|\bv^\dag\bM^{-1}\bv_t\right|^2}{(\bv^\dag\bM^{-1}\bv)(\bv_t^\dag\bM^{-1}\bv_t)}.
\end{equation}
Observe that in \eqref{costheta}, $\cos^2\phi = 1$ represents the perfectly matched case addressed in the previous
subsection, whereas $\cos^2\phi = 0$ corresponds to the totally mismatched scenario, namely $\bv_t$
is orthogonal to $\bv$ in the whitened observation space.
\begin{figure}[t]
  \centering
  \subfigure[$N=8$, $K=16$]{\includegraphics[width=7.7cm]{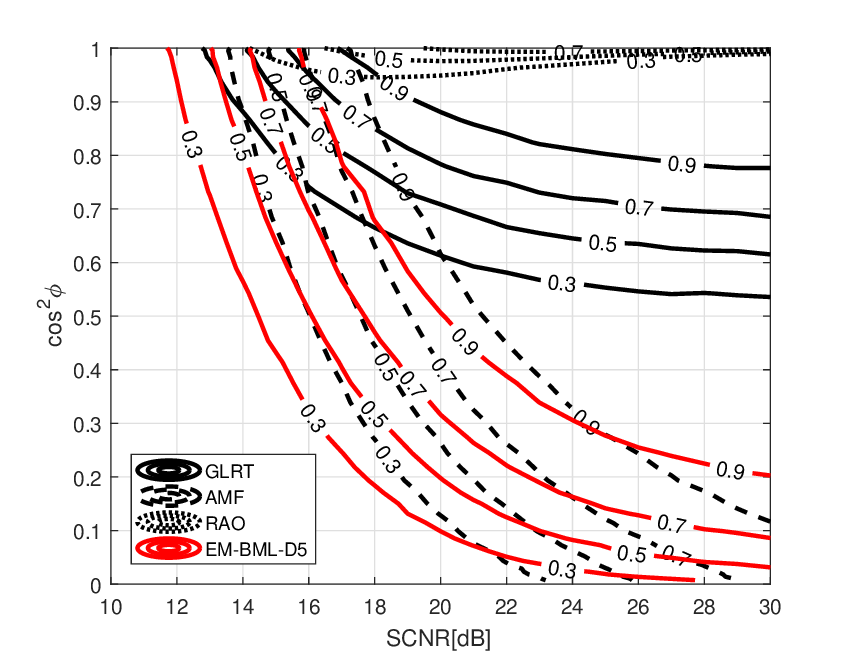}\label{fig-mis-N16K32}}
  \subfigure[$N=16$, $K=32$]{\includegraphics[width=7.7cm]{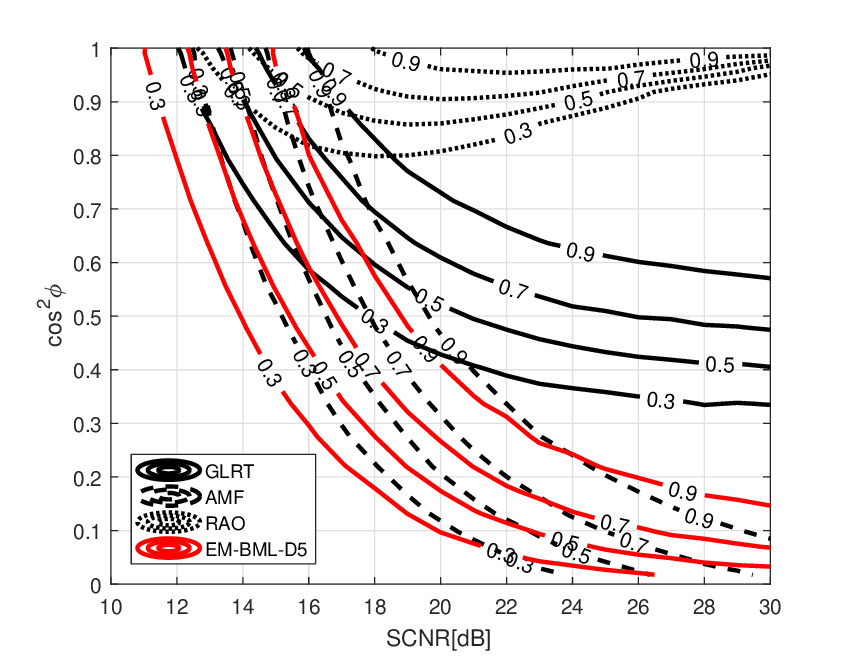}\label{fig-mis-N16K24}}
  \caption{Contours of constant $P_d$ for mismatched targets, $\text{CNR}=30$ dB.}\label{fig-mis}
\end{figure}

Fig. \ref{fig-mis} contains the contours of constant $P_d$ as a function of the $\cos^2\phi$ and SCNR.
Inspection of the figure shows that the {\color{black}EM-BML-D} proposed in this paper exhibits a similar performance
in terms of rejection capability of mismatched targets as the AMF.
In fact, in the low SCNR region,
the {\color{black}EM-BML-D} exhibits the lowest selectivity.
In the high SCNR region, the {\color{black}EM-BML-D} rejects mismatched signals slightly better than the AMF.
Otherwise stated, the {\color{black}EM-BML-D} is the most robust detector in the presence of target mismatch.
{\color{black}In Fig. \ref{fig-mis-cos}, we analyze the performance in the case of mismatched signals from another perspective
by plotting the curves of $P_d$ versus $\text{cos}^2\phi$ for a given value of SCNR such that all of the detectors
exhibit $P_d=1$ at $\text{cos}^2\phi=1$ (except for Rao test).
In fact, the $P_d$ of Rao test in Fig. \ref{fig-mis-N16K32-30} is 0 when $\text{cos}^2\phi<0.9$ and equals to 0.84 at $\text{cos}^2\phi=1$.
Notice that Fig. \ref{fig-mis-cos} is actually a slice of Fig. \ref{fig-mis} at SCNR $=25$ dB.
From the inspection of the figure, it turns out that the EM-BML-D has a robustness that is similar to
that of the AMF while is less selective than the GLRT and Rao.
}
\begin{figure}[h]
  \centering
  \subfigure[$N=8$, $K=16$, SCNR $=25$ dB]{\includegraphics[width=7.7cm]{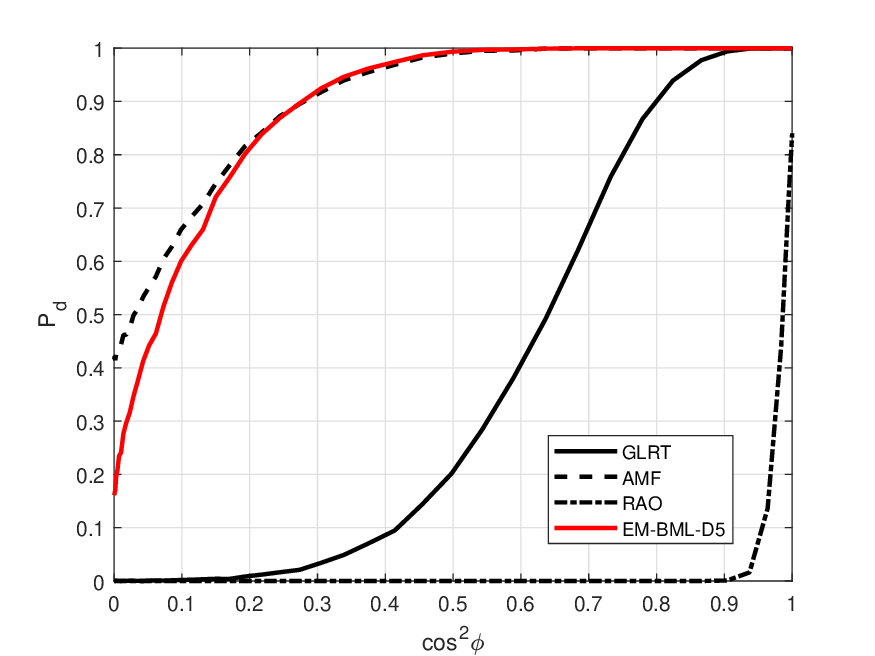}\label{fig-mis-N16K32-30}}
  \subfigure[$N=16$, $K=32$, SCNR $=25$ dB]{\includegraphics[width=7.7cm]{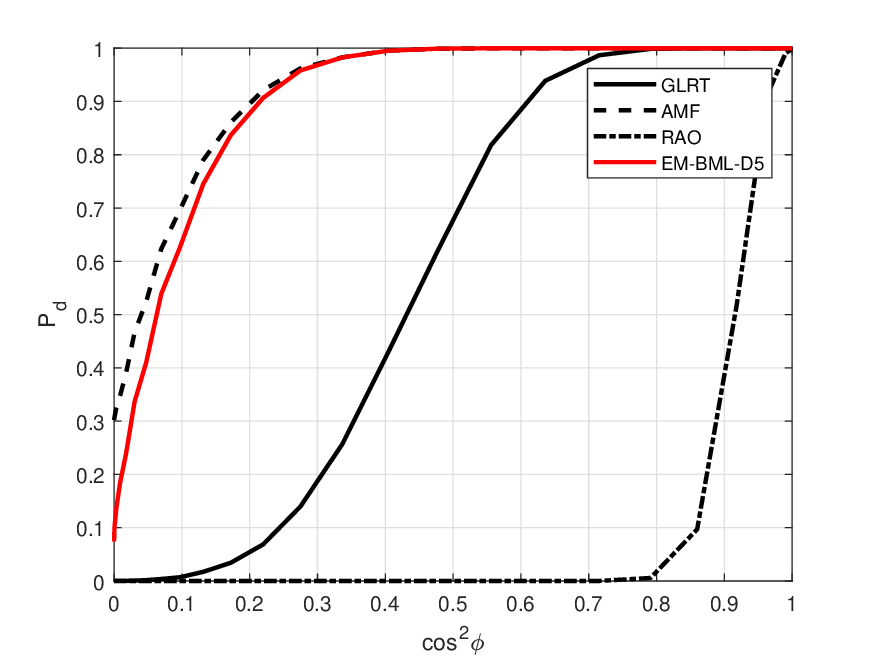}\label{fig-mis-N16K24-20}}
  \caption{$P_d$ versus $\text{cos}^2\phi$ for constant SCNR, $\text{CNR}=30$ dB.}\label{fig-mis-cos}
\end{figure}

\section{Conclusion}

In this paper, we have proposed a new radar target detection architecture grounded on the Bayesian framework and
by exploiting the estimates from the EM algorithm in the presence of latent random variables.
Precisely, the posterior information provided by the recursive procedure of the EM algorithm allows for the evaluation
of the responsibility that the hypotheses take in ``explaining'' the data, and thus the decision for the presence of a potential target.
At the design stage, we have modeled distribution of the collected data as a combination of the distributions under the two hypotheses,
and we have used the estimates provided by the EM algorithm to build up a Bayesian detector. More importantly, we have proved that the proposed
detector ensures the desired CFAR property with respect to the disturbance covariance matrix.
The illustrative examples have shown that the proposed {\color{black}EM-BML-D} significantly outperforms the
considered conventional counterparts in terms of
matched target detection and is robust to mismatched signals.

Future research tracks may encompass the design of selective schemes based on the EM algorithm
and the Bayesian framework. The generalization of {\color{black}EM-BML-D} towards multiple alternative hypothesis test also
represents an interesting direction for the next works.

\bibliographystyle{IEEEtran}
\bibliography{group_bib,group_bib_dan}

%
%

\end{document}